\documentclass[smallcondensed]{svjour3}     % onecolumn (ditto)
\smartqed  % flush right qed marks, e.g. at end of proof

\usepackage{graphicx}%\graphicspath{{./Figures/}}
\usepackage{amsmath,amssymb}
\usepackage{siunitx}
\usepackage{ifpdf}

\usepackage{color}

\renewcommand{\a}{\mathrm{a}}
\renewcommand{\d}{\mathrm{d}}
\newcommand{\e}{\mathrm{e}}
\newcommand{\f}{\mathrm{f}}

\newcommand{\m}{\mathrm{m}}
\renewcommand{\r}{\mathrm{r}}

\newcommand{\dd}{\partial}
\newcommand{\RR}{\mathbb{R}}
\newcommand{\III}{\mathcal{I}}
\newcommand{\SSS}{\mathcal{S}}
\newcommand{\TTT}{\mathcal{T}}

\newcommand{\coloneqq}{\mathrel{\mathop:}=}

\newcommand{\nspec}{n_{\mathrm{spec}}}
\newcommand{\nelem}{n_{\mathrm{elem}}}
\newcommand{\nreac}{n_{\mathrm{reac}}}
\newcommand{\species}[1]{\mathrm{#1}}

\journalname{}

\begin{document}

\title{An optimization approach to kinetic model reduction for combustion chemistry}
%\subtitle{Do you have a subtitle?}
\titlerunning{An optimization approach to kinetic model reduction for combustion chemistry}

\author{Dirk Lebiedz \and Jochen Siehr}
%\authorrunning{Dirk Lebiedz, Jochen Siehr}

\institute{
D. Lebiedz\at
Center for Systems Biology (ZBSA),
University of Freiburg,\\
Habsburgerstra\ss{}e 49,
79104 Freiburg im Breisgau, Germany\\
%Tel.: +49 761 203 97161\\
%Fax:  +49 761 203 97186\\
%\email{dirk.lebiedz@biologie.uni-freiburg.de}\\
\emph{Present address:}\\
  Institute for Numerical Mathematics, University of Ulm, \\
  Helmholtzstra\ss{}e 20, 89081 Ulm, Germany\\
  \email{dirk.lebiedz@uni-ulm.de}
\and
J. Siehr\at
Interdisciplinary Center for Scientific Computing (IWR),
University of Heidelberg,\\
Im Neuenheimer Feld 368,
69120 Heidelberg, Germany\\
%Tel.: +49 761 203 97160\\
%Fax:  +49 761 203 97186\\
%\email{jochen.siehr@iwr.uni-heidelberg.de}\\
\emph{Present address:}\\
  Institute for Numerical Mathematics, University of Ulm, \\
  Helmholtzstra\ss{}e 20, 89081 Ulm, Germany\\
  \email{jochen.siehr@uni-ulm.de}
}

\date{}

\maketitle

% This manuscript is aimed for the special issue ``Flow and Combustion in modern
% Gas Turbine Combustion Chambers'' (NWRF SI). \newline

\begin{abstract} % 150-250 words
Model reduction methods are relevant when the computation time of a full
convection--diffusion--reaction simulation based on detailed chemical reaction
mechanisms is too large. In this article, we review a model reduction approach
based on optimization of trajectories and show its applicability to realistic
combustion models. As most model reduction methods, it identifies points on a
slow invariant manifold based on time scale separation in the dynamics of the
reaction system. The numerical approximation of points on the manifold is
achieved by solving a semi-infinite optimization problem, where the dynamics
enter the problem as constraints. The proof of existence of a solution for an
arbitrarily chosen dimension of the reduced model (slow manifold) is
extended to the case of realistic combustion models including thermochemistry
by considering the properties of proper maps.
The model reduction approach is finally applied to three models based on
realistic reaction mechanisms: 1.\ ozone decomposition as a small test case;
2.\ simplified hydrogen combustion for comparison with another model reduction
method;
3.\ syngas combustion as a test case including all features of a detailed combustion
mechanism.

 % 4-6 keywords
\end{abstract}

% ---------------------------------------------------------------------------- %
\section{Introduction}\label{intro}
The modeling of chemically reacting flows comprises the interplay between flow
(convection), diffusion, and chemical reaction processes. This
interplay is fairly complex if the model is based on a detailed combustion
mechanism involving a large number of chemically reactive species and reactions.
Here complexity reduction and model reduction methods can be effective tools.

A common aim of many model reduction approaches is the identification (and
computation) of so called slow invariant manifolds (SIM). Many model reduction
methods are applied to the chemical source term of the system of reaction
transport partial differential equations (PDE), which describe the reactive
flow. This means, the model reduction method is only regarding a system of
ordinary differential equations (ODE) modeling the kinetic source term.
Trajectories in chemical composition space are relaxing to the SIM while
converging towards equilibrium. In this sense, the SIM represent the slow
chemistry for a time scale separation between the tangent and normal dynamics. The
existence of a SIM is closely related to multiple time scales and time scale
separation and a spectral gap in the eigenvalues of the Jacobian of the
chemical source term characterizes the ratio of contraction rates in tangent
and normal directions.

In general, certain species will be chosen as ``represented'' ones for the
simulation of the reacting flow based on reduced chemistry models; these are
also called \emph{reaction progress variables}. These variables parametrize the reduced
chemistry model and for these variables the transport PDE are actually
solved. The values of the remaining unrepresented variables are computed in
dependence of the represented species by considering a point on the slow
manifold parameterized by the reaction progress variables.

Historically, model reduction techniques have been used since the quasi steady
state assumption (QSSA) and the partial equilibrium assumption (PEA) became
popular---methods that had to be performed by hand \cite{Warnatz2006}. By
contrast, modern model reduction methods compute a slow manifold approximation
automatically without the need of expert knowledge for identification of
reactions in partial equilibrium and species in steady state within a complex 
chemical reaction mechanism. A very popular automatic method is the intrinsic
low dimensional manifold (ILDM) method, that has originally been published by
Maas and Pope in \cite{Maas1992} in 1992, and its extensions. Also the
computational singular perturbation (CSP) method by Lam and Goussis
\cite{Lam1985,Lam1994} is widely applied, e.g.\ in \cite{Najm2010}. Another
widely applied method, e.g.\ in \cite{Ketelheun2011}, is the method of flamelet
generated manifolds \cite{Delhaye2008}. For an overview of model reduction
methods, see the review \cite{Gorban2005} and the references it contains.

The scope of this manuscript is a guidance to the application of the
optimization based model reduction method as introduced in \cite{Lebiedz2004c}
and future developments. The method is applied to chemical combustion
mechanisms and results are discussed.
The outline of this manuscript is as follows. In Section~\ref{s:chemistry},
the chemistry models regarded here are presented. The optimization problem for
identification of the SIM is explained in Section~\ref{s:variational_problem}.
A short overview of appropriate numerical
methods needed for solving the optimization problem is given in
Section~\ref{s:numerics}. Results of an application to three models are shown
and discussed in Section~\ref{s:results}.

% ---------------------------------------------------------------------------- %
\section{Model equations in combustion chemistry}\label{s:chemistry}
The model reduction approach discussed here is applied only to the reaction part
of the reactive flow model. In this section we review knowledge that can be
found in text books as e.g.\ \cite{Kee2003,Warnatz2006}.

\subsection{Mass conservation laws}\label{ss:conservation}
The general reaction transport equation in the variable $\zeta$ which can be
mass fractions, temperature, or any variable describing the state of the system 
depending on time $t$ and space $x$ can be written as
\begin{equation}
 \dd_t \zeta = \SSS(\zeta) + \TTT(\zeta,\dd_x \zeta,\dd_x^2\zeta),
\end{equation}
where $\SSS$ is the (chemical) source term and $\TTT$ the physical transport
operator, i.e.\ convection and diffusion.

The model comprises $\nspec$ chemical species composed by $\nelem$ chemical
elements, and the chemical source term $\SSS$ obeys the law of elemental mass
conservation and an energetic balance which will be regarded in
Section~\ref{ss:energy_balance}. In the following, the variables $z_i$ are given
in terms of specific moles, which are defined as the amount of species $i$
($n_i$) divided by the total mass ($m$) of the system, which is the same as the
species's mass fractions ($w_i$) divided by its molar mass ($M_i$):
\begin{equation*}
 z_i = \frac{n_i}{m} = \frac{w_i}{M_i}.
\end{equation*}

In these variables, the mass conservation of each element in the system is
formulated as
\begin{equation}\label{eq:conserv_mass}
 \bar{z}_i = \sum_{j=1}^{\nspec}\chi_{ij}z_j \quad i=1,\dots,\nelem,
\end{equation}
where $\chi_{ij}$ is the atomic composition coefficient---the number of element
$i$ in species $j$.
There is also a restriction to the choice of the elemental specific moles
$\bar{z}_i$ requiring that the mass fractions sum to one
\begin{equation*}
 \sum_{i=1}^{\nelem} \bar{M}_i\bar{z}_i = 1,
\end{equation*}
where $\bar{M}_i$ is the molar mass of element $i$. This is equivalent to the
conservation of the total mass of the system.

\subsection{Energetic balance}\label{ss:energy_balance}
Energy conservation has to be regarded, too. We consider systems within one of
the four standard thermodynamic environments, i.e.
\begin{itemize}
 \item isothermal and isochoric
 \item isothermal and isobaric
 \item adiabatic and isochoric (hence isoenergetic)
 \item adiabatic and isobaric (hence isenthalpic)
\end{itemize}
systems. Only the temperature is fixed in the isothermal case, while the
specific enthalpy $h$
\begin{equation}\label{eq:fix_h}
 h = \sum_{i=1}^{\nspec} \bar{H}_i^{\circ}(T)\;z_i,
\end{equation}
or the specific internal energy $e$
\begin{equation}\label{eq:fix_e}
 e = h - \frac{RT}{\bar{M}},
\end{equation}
respectively, are fixed in the adiabatic cases, where we assume an ideal gas
mixture, and $\bar{M}$ is
the mean molar mass
\begin{equation*}
 \bar{M} = \frac{1}{\sum_{i=1}^{\nspec}z_i}.
\end{equation*}
The molar enthalpy $\bar{H}_i^{\circ}(T)$ of species $i$ is computed by
evaluation of so called NASA polynomials \cite{Burcat2005}.

\subsection{Standard systems}\label{ss:systems}
In the optimization problem discussed later in Section
\ref{s:variational_problem}, the dynamics of the system are considered as
constraints. Therefore, the reaction ODE system that is given by the source term
only is discussed in the following. This ODE system includes mass action
kinetics and incorporates a differential form of the elemental mass conservation
laws.

The mass balance in specific moles $z_i$ can be formulated as
\begin{equation}\label{eq:source_mass}
 \d_t z_i \coloneqq \frac{\d}{\d t}z_i = S^{\m} \coloneqq \frac{\omega}{\rho}, \quad i=1,\dots,\nspec.
\end{equation}
In the right hand side of Equation (\ref{eq:source_mass}), the symbol $\rho$
refers to the overall mass density in the system which is given via
\begin{equation*}
 \rho =  \frac{m}{V} = \frac{p \bar{M}}{RT}.
\end{equation*}
In the isochoric case, the total mass and volume $V$ are to be known; in the
isobaric case the total pressure $p$, the gas constant $R$, the temperature $T$,
and the mean molar mass are necessary. The molar net chemical production rate
is denoted by $\omega$ in Eq.~(\ref{eq:source_mass}). It has to be computed
based on a set of chemical elementary reactions and their
parameters as described in Section~\ref{ss:kinetics}.

In our case, we formulate the energy balance via the right hand side of the
temperature equation $\d_t T = \frac{\d}{\d t}T = S^{\e}$. In the isothermal
case, we have $S^{\e}=0$ of course. In the adiabatic
cases, energy or enthalpy conservation define the concise form of $S^{\e}$.

\subsection{Chemical kinetics}\label{ss:kinetics}
In the remainder of this section, we review the computation of $\omega$. A
chemical combustion mechanism generally is given as a set of $\nreac$ elementary
reactions involving $\nspec$ species (and eventually a third body $\species{M}$)
\begin{equation*}
\sum_{i=1}^{\nspec} \nu_{ij}^{\prime} \species{X}_i \rightleftharpoons \sum_{i=1}^{\nspec} \nu_{ij}^{\prime\prime} \species{X}_i,
\quad j=1,\dots,\nreac
\end{equation*}
with the chemical species $\species{X}_i$ and the forward and reverse
stoichiometric coefficients $\nu_{ij}^{\prime}$ and $\nu_{ij}^{\prime\prime}$.
The forward and reverse rate of reaction $j$ is given via
\begin{equation}\label{eq:rates}
\begin{aligned}
 r_{\f,j} &= k_{\f,j} \prod_{i=1}^{\nspec} c_i^{\nu_{ij}^{\prime}}\\
 r_{\r,j} &= k_{\r,j} \prod_{i=1}^{\nspec} c_i^{\nu_{ij}^{\prime\prime}}
\end{aligned}
\end{equation}
with the concentrations $c_i$ of species $i$ and the rate constants
$k_{\f,j}$ and $k_{\r,j}$. Using the net stoichiometric coefficient
\begin{equation*}
 \nu_{ij} = \nu_{ij}^{\prime\prime} - \nu_{ij}^{\prime},
\end{equation*}
the net rate of change $\omega_i$ of species $i$ is computed as
\begin{equation}\label{eq:net_prod}
 \omega_i = \sum_{j=1}^{\nreac} \nu_{ij} (r_{\f,j} - r_{\r,j}).
\end{equation}
In case a third body $\species{M}$ takes part in reaction $j$, third body
collision efficiencies $\alpha_i$ for all species $i=1,\dots,\nspec$ are to be
given. The third body concentration
\begin{equation}\label{eq:tb_conc}
 c_{\species{M}} = \sum_{i=1}^{\nspec} \alpha_i c_i
\end{equation}
is then multiplied to the products in Equation (\ref{eq:rates}).

Formulas for the computation of the rate coefficients are stated in the
following. The elementary reactions in the mechanism considered here are given
in Arrhenius format and pressure dependent Troe format, resp. The three
parameters $A$, $b$, and $E_{\a}$ are given in the Arrhenius kinetics for each
reaction.
The forward reaction rate coefficient is computed via the extended Arrhenius
formula
\begin{equation}\label{eq:Arrhenius}
 k_{\f,j} = A \ T^{\tfrac{b}{1 \mathrm{K}}} \ \e^{-\tfrac{E_{\a}}{RT}}.
\end{equation}

A more complicated formula applies in case of pressure dependent reactions.
Here $k_{\f,j}$ is computed using Troe fall off curves
\cite{Gilbert1983,Troe1983}.
Both the low pressure rate coefficient $k_{0}$ and the high pressure
coefficient $k_{\infty}$ are given via the extended Arrhenius formula
(\ref{eq:Arrhenius}). These are used together with a third body $\species{M}$ to
compute the reduced pressure
\begin{equation*}
 p_{\r} = \frac{k_{0} c_{\species{M}}}{k_{\infty}},
\end{equation*}
where $c_{\species{M}}$ is defined as in Equation (\ref{eq:tb_conc}). With this parameter,
the final rate constant is computed as
\begin{equation*}
 k_{\f,j} = k_{\infty} \frac{p_{\r}}{1 + p_{\r}} F
\end{equation*}
with a function $F$. To compute the value of $F$, Gilbert, Luther, and Troe
\cite{Gilbert1983,Troe1983} introduced the formula
\begin{equation*}
 \lg F = \left\{1+ \left[\frac{\lg p_{\r} + c}{n - d(\lg p_{\r} + c )} \right]^2\right\}^{-1} \lg F_{\textrm{c}}
\end{equation*}
with a set of simplifications
\begin{align*}
 c &= -0.4 - 0.67 \lg F_{\textrm{c}} \\
 n &= 0.75 - 1.27 \lg F_{\textrm{c}}\\
 d &= 0.14
\end{align*}
and the F-center-value
\begin{equation*}
 F_{\textrm{c}} = (1-a) \exp\left(-\frac{T}{T^{***}}\right)
 + a \exp\left(-\frac{T}{T^{*}}\right)
 + \exp\left(-\frac{T^{**}}{T}\right)
\end{equation*}
which includes the four parameters $a$, $T^{*}$, $T^{**}$, and $T^{***}$. These
are given for each so called Troe-reaction.

The reverse reaction rate constant $k_{\f,j}$ of reaction $j$ is computed via
the equilibrium constant $K_{c,j}$ of the reaction.\footnote{
 In some publications, the reverse rate coefficients of Arrhenius type reactions
 are computed with fitted Arrhenius parameters for the reverse reactions. We
 also used this strategy in previous publications as
 e.g.\ \cite{Lebiedz2010,Lebiedz2011}. However, this is impossible in case of
 pressure dependent reactions and it may lead to inconsistent values of
 thermodynamic quantities in the mass action kinetics in relation to the heat of
 reaction. Therefore, we prefer the thermodynamic approach in this
 manuscript.}
The equilibrium constant of reaction
$j$ in terms of concentrations is given as
\begin{equation*}
 K_{c,j} = \left(\frac{p^{\circ}}{RT}\right)^{\nu_j}\exp\left(\frac{\Delta S_{\r,j}^{\circ}}{R}
           - \frac{\Delta H_{\r,j}^{\circ}}{RT}\right)
\end{equation*}
with the standard pressure $p^{\circ}$ and the net change of the number of
species present in the gas phase
\begin{equation*}
 \nu_j = \sum_{i=1}^{\nspec} \nu_{ij}.
\end{equation*}
The change in entropy $\Delta S_{\r,j}$ and enthalpy $\Delta H_{\r,j}$ can be
computed by using an evaluation of the NASA polynomials for their molar
values of the species involved. The reverse rate of reaction $j$ finally is
\begin{equation*}
 k_{\r,j} = \frac{k_{\f,j}}{K_{c,j}}.
\end{equation*}

% ---------------------------------------------------------------------------- %
\section{Optimization problem}\label{s:variational_problem}
In 2004, Lebiedz introduced a model reduction method based on the minimization
of entropy production rate along trajectories in chemical composition space
\cite{Lebiedz2004c}. The basic idea is that the SIM is characterized by
maximum relaxation of the system dynamics under given constraints of fixed
reaction progress variables. This approach has been extended and refined in a
number of following publications
\cite{Lebiedz2010a,Lebiedz2010,Lebiedz2011a,Reinhardt2008}.

The general geometric situation in the phase space of the reaction system
spanned by the variables $z_i$ can be seen in the sketch in Figure \ref{f:mani}.
Trajectories bundle on the manifolds of slow motion, that are hierarchically
ordered.
\begin{figure}[htpb]
 \centering
 \includegraphics[width=0.8\textwidth]{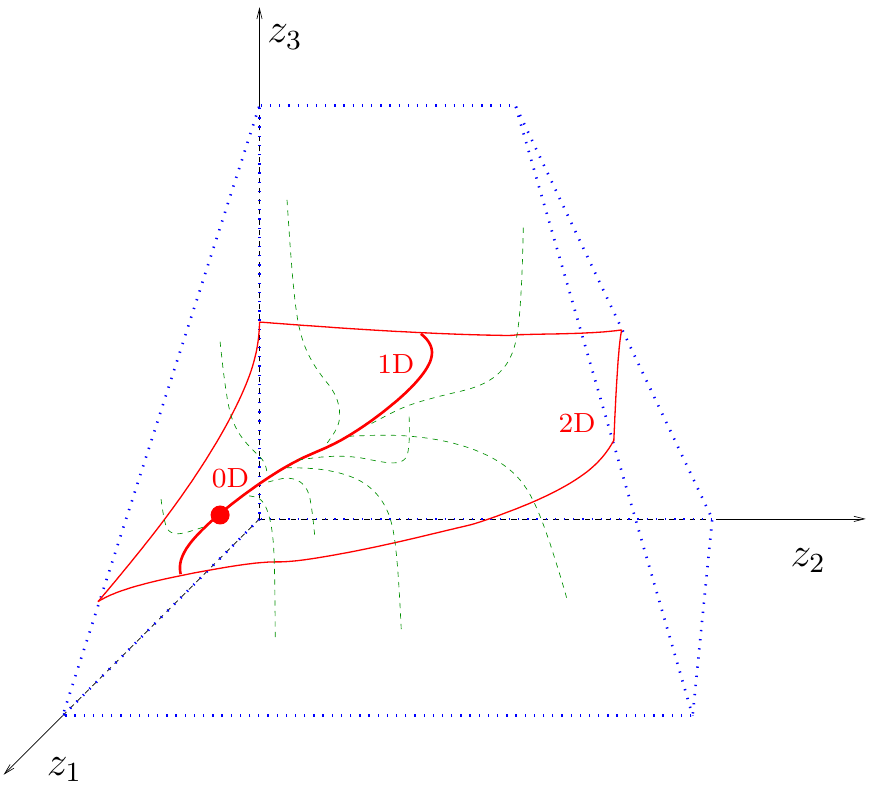}
%   \ifpdf
%    \input{./Figures/manifold1.pdf_t}
%   \else
%    \input{./Figures/manifold1.pstex_t}
%   \fi
 \caption[Sketch of manifold]{\label{f:mani}Sketch of the chemical composition
 space. The domain, where the
 dynamics take place is the blue-bounded polytope, here three-dimensional.
 Within this polytope there might be (depending on a possible time scale
 separation) a two-dimensional manifold (depicted in red), where
 (shown in green) trajectories relax onto. The trajectories relax
 then onto the one-dimensional manifold within the two-dimensional one. Finally
 the trajectories converge toward the zero-dimensional manifold:
 the equilibrium.}
\end{figure}
The aim is to compute an approximation of such a manifold of given dimension
point-wise such that the free variables are computed depending on the given
reaction progress variables which parametrize the manifold.

Following the idea of
\cite{Lebiedz2004c,Lebiedz2010a,Lebiedz2010,Lebiedz2011a,Reinhardt2008},
an optimization problem has to be solved for the approximation of points on the
manifold. It can be written in specific moles and temperature as minimization
of an objective functional
\begin{subequations}\label{eq:op}
 \begin{equation} \label{eq:op:obj}
  \min_{z,T}  \int_{t_0}^{t_{\textrm f}}  \Phi(z(t)) \; \d t
 \end{equation}
 subject to
 \begin{align}
   \d_t z(t) &= S^{\m}(z(t),T(t)) \label{eq:op:dyn_m}\\
   \d_t T(t) &= S^{\e}(z(t),T(t)) \label{eq:op:dyn_e}\\
           0 &= C\left(z(t_{*}),T(t_{*})\right) \label{eq:op:con}\\
           0 &= z_j(t_{*}) - z_j^{t_{*}},\quad j \in \III_{\mathrm{rpv}} \label{eq:op:pv}\\
           0 &\leqslant z(t),T(t) \label{eq:op:pos}
 \end{align}
 and
 \begin{align}
  t     &\in [t_0,t_{\textrm{\upshape f}}] \\
  t_{*} &\in [t_0,t_{\textrm{\upshape f}}]\quad \text{\upshape (fixed)}.
 \end{align}
\end{subequations}
In the following, we explain optimization problem (\ref{eq:op}) in detail
starting with the constraints.

\paragraph{Chemical source and heat of reaction}
For the numerical solution of the optimization problem, the dynamics of the
system have to be computed. The dynamics are given via the ODE system in the
constraints (\ref{eq:op:dyn_m}) and (\ref{eq:op:dyn_e}). This ensures to
identify as solution of (\ref{eq:op}) a special solution trajectory piece.

\paragraph{Conservation relations}
All necessary additional conservation laws are combined in the (nonlinear)
function $C$. As discussed before, the dynamics (\ref{eq:op:dyn_m})
and (\ref{eq:op:dyn_e}) contain differential forms of the balances of mass and
energy. The concise values of the conserved quantities have to be specified
at some (fixed) point in time along a solution which we choose to be $t_*$. This
is Equation~(\ref{eq:conserv_mass}) and a specification of a fixed temperature,
enthalpy, or energy, as e.g.\ (\ref{eq:fix_h}) or (\ref{eq:fix_e}), depending
on the assumed thermodynamic environment.

\paragraph{SIM parametrization}
In order to approximate the SIM, a parametrization needs to be specified.
The species (i.e.\ the specific moles $z_i$) which serve as reaction progress variables
and especially their number have to be specified in advance. The number of
progress variables determines the chosen dimension of the SIM to be
approximated.

In the case illustrated in Figure \ref{f:mani}, one might choose $z_1$ and
$z_2$ as reaction progress variables for parametrization in order to compute a
value for the remaining free variable $z_3$ which is supposed to be on the
two-dimensional manifold. Alternatively, one can choose only $z_1$ as reaction
progress variable for approximation of the one-dimensional manifold.

The indices of the reaction progress variables are collected in the index set
$\III_{\text{rpv}}\subset\{1,\dots,\nspec\}$, and their values are fixed in the
optimization problem to $z_j^{t_{*}}$.

\paragraph{Positivity}
Since only positive values of specific moles and temperature have a physical
meaning, this is included in (\ref{eq:op:pos}). In the optimization context
with a realistic combustion mechanism included in the constraints, this is also
technically important as negative values of $z_i$ and $T$ can result
in undefined values (logarithm of a negative number) of the right hand sides
$S^{\m}$ and $S^{\e}$. The positivity and the linear mass conservation relations
define a polytope as depicted in Figure \ref{f:mani}.

\bigskip

The chosen point in time $t_{*} \in [t_0,t_\f]$ specifies the position where the
fixation of the reaction progress variables and the constraint $C$ is applied
along the trajectory piece, which is optimized. In first publications,
e.g.\ \cite{Lebiedz2004c,Lebiedz2010},
$t_{*}=t_0$ is chosen. This incorporates the demand that the trajectories are
fully relaxed to the SIM at time $t_{*}$ and no further relaxation takes place
afterward.

The inverse idea is the fixation at the end point $t_{*}=t_\f$. The solution
point and solution trajectory piece of the optimization problem (\ref{eq:op})
is supposed to be part of the SIM. Therefore, also in backward direction of time
the trajectory is supposed to be already relaxed.

This is related to the definition of positive and negative invariance of a set
under a flow in dynamical systems theory.
\begin{definition}[Invariant set, \cite{Wiggins1996}]
 Let $\Omega\subset\RR^n$ be a set. It is called \emph{invariant} under the
 vector field $\dot{\zeta} = \SSS(\zeta)$, $\zeta\in\RR^n$ if for any
 $\zeta_0\in\Omega$ it holds that $\zeta(t;\zeta_0)\in\Omega$ for all $t\in\RR$,
 where $\zeta(t;\zeta_0)$ denotes the solution of the initial value problem
 $\dot{\zeta} = \SSS(\zeta)$ with initial value $\zeta_0$ at $t=0$.
 It is called \emph{positively invariant} if this conditions holds for positive
 $t\geqslant0$ and \emph{negatively invariant} if the conditions holds for
 negative $t\leqslant0$.
\end{definition}

Clearly, the essential degrees of freedom of the optimization problem is the
effective phase space dimension $\nspec - \nelem$ minus the number of
reaction progress variables $|\III_{\mathrm{rpv}}|$. The goal of solving the
optimization problem (\ref{eq:op}) is the determination (species reconstruction)
of the ``missing'' values $z_i(t_{*}),\ i \notin \III_{\mathrm{rpv}}$, as a
function of the parameters $z_j^{t_{*}},\ j \in \III_{\mathrm{rpv}}$.

\subsection{The objective functional}\label{ss:criterion}
The relaxation criterion $\Phi$ is supposed to measure the degree of chemical
force relaxation along a trajectory. Several criteria have been tested for their
SIM approximation quality, especially in \cite{Lebiedz2010}.

The SIM to be approximated is considered to be slow. This means, the residence
time of the trajectory in some open neighborhood of a point on the SIM should
be large---conversely the change and (to second order) the rate of change of the
variable values is supposed to be small. A similar idea has been pointed out
already in \cite{Girimaji1998}. The rate of change is closely related
to the curvature of the trajectories as geometrical objects in phase space.
The rate of change of the specific moles is simply $\dot{z} = S^{\m}$. Its rate
of change is the second derivative
\begin{equation*}
 \ddot{z}(t) = J_{S^{\m}}(z(t),T(t))\;S^{\m}(z(t),T(t)),
\end{equation*}
where we denote the Jacobian of a function $S$ as $J_S$.
This can be seen as a directional derivative of the chemical source w.r.t.\ its
own direction $v$ which should be regarded normalized
$v\coloneqq\frac{\dot{z}}{\|\dot{z}\|_2}=\frac{S^{\m}}{\|S^{\m}\|_2}$
\begin{equation*}
 D_v \dot{z}(t) \coloneqq \frac{\d}{\d \alpha} S^{\m}(z(t) + \alpha v,T(t))
  \Big|_{\alpha=0}  = J_{S^{\m}}(z(t),T(t)) \frac{S^{\m}(z(t),T(t))}{\|S^{\m}(z(t),T(t))\|_2},
\end{equation*}
where $\|\cdot\|_2$ denotes the Euclidean norm.
The evaluation of this expression within the integral should be done in arc
length. A re-parametrization cancels out the norm $\|S^{\m}(z(t),T(t))\|_2$
such that (in notation that coincides with the general problem (\ref{eq:op}))
a reasonable candidate for the criterion $\Phi$ would be
\begin{equation}
 \Phi(z(t)) = \| J_{S^{\m}}(z(t),T(t))\;S^{\m}(z(t),T(t)) \|_2^2.
\end{equation}

\subsection{Solution of the optimization problem}\label{ss:sol_theory}
In \cite{Lebiedz2011a}, the authors study theoretical properties of the
optimization based model reduction method as described in the sections before.
It is shown there always exists a solution of the optimization problem
(\ref{eq:op}) with only linear (mass conservation) constraints if there exists a
feasible solution.

In case of a realistic combustion mechanism as a model for an adiabatic
system, the nonlinear internal energy conservation or enthalpy conservation comes
into play. The existence proof will be extended to this cases in the following.
The crucial point is the compactness of the feasible domain, which is more
complicated to ensure in the nonlinear case.

A simple way to guarantee compactness of the feasible domain would be an upper
bound for the temperature. Together with the compactness argument for the linear
constraints \cite{Lebiedz2011a}, the compactness of the feasible domain is
obvious. But it is not clear at all where to choose the upper cut off for the
temperature.

We avoid this temperature cut off but make use of the definition of molar
enthalpy via NASA polynomials. Thereby we accept the temperature to outrange the
domain where the NASA polynomials approximate the molar enthalpy of the species
appropriately.

The specific enthalpy $h$ of a system is given as
\begin{equation} \label{eq:entropy}
 h = \sum_{i=1}^{\nspec} \bar{H}_i^{\circ}(T)\;z_i.
\end{equation}
The equation for the specific internal energy $e$ is
\begin{equation} \label{eq:energy}
 e = h - RT\sum_{i=1}^{\nspec} z_i.
\end{equation}
The molar enthalpy $\bar{H}_i^{\circ}(T)$ of species $i$ is a continuous
function in the temperature $T$. In our case, it is computed by evaluation of the
NASA polynomials. Their formula is
\begin{equation} \label{eq:nasa}
 \frac{\bar{H}_i^{\circ}(T)}{R} = a_6 + a_1 T + \frac{a_2}{2}T^2 + \frac{a_3}{3}T^3 + \frac{a_4}{4}T^4 + \frac{a_5}{5}T^5
\end{equation}
with two sets of coefficients $a_i$, $i=1,\dots,6$. One set is given for a
temperature lower than a certain switch temperature $T < T_{\mathrm{sw}}$ and
one set of coefficients for high temperature $T \geqslant T_{\mathrm{sw}}$.
The two branches are connected at
$T_{\mathrm{sw}}$ at least continuously. There are also upper and lower bounds
for the temperature, where the polynomial approximation is valid. We ignore
these bounds for the following theory.

\begin{definition}[Proper map, \cite{Lee2011}]
 Let $X$ and $Y$ be topological spaces. A map (continuous or not)
 $H : X \rightarrow Y$
 is called \emph{proper} if the preimage $H^{-1}(K)$ of each compact subset
 $K\subset Y$ is compact.
\end{definition}
To formulate a sufficient condition for properness we need the
\begin{definition}[Divergence to infinity, \cite{Lee2011}]
 If $X$ is a topological space, a sequence $(x_n)$ in $X$ is said to
 \emph{diverge to infinity} if for every compact set $K \subseteq X$ there are
 at most finitely many indices $n$ with element $x_n \in K$.
\end{definition}

A sufficient condition for properness is the following
\begin{lemma}[Properness condition, \cite{Lee2011}]
 Suppose $X$ and $Y$ are topological spaces, and $H : X \rightarrow Y$ is a
 continuous map. If X is a second countable Hausdorff space and $F$ takes
 sequences diverging to infinity in $X$ to sequences diverging to infinity in
 $Y$, then $F$ is proper.
\end{lemma}
\begin{proof}
 See \cite[p.~119]{Lee2011}.
\end{proof}

\begin{lemma}[Properness of $h$ and $e$]\label{lemma:proper}
 The specific enthalpy $h$ and the specific internal energy $e$ defined via
 NASA polynomials seen as functions in $T$ and $z$ are proper
 maps.
\end{lemma}

\begin{proof}
 The vector space $\RR^{\nspec+1}$ is a second countable Hausdorff space. Any
 non-constant polynomial takes sequences diverging to infinity in $\RR^n$
 equipped with its Euclidean metric induced topology to
 sequences diverging to infinity in $\RR$. We can see $h$ and $e$ as polynomials
 of sixth degree in $z_i$ and $T$, see Eq.~(\ref{eq:entropy}),
 (\ref{eq:energy}), and (\ref{eq:nasa}).
 Therefore, $h:\RR^{\nspec+1}\rightarrow\RR$ and
 $e:\RR^{\nspec+1}\rightarrow\RR$ are proper maps.\qed
\end{proof}

Using this information, we can extend the existence lemma~2.1 in
\cite{Lebiedz2011a}.
\begin{lemma}\label{lemma:compact}
 The feasible set at $t_*$
 \begin{equation*}
 M = \{(z,T): C(z,T)=0;\ z_j - z_j^{t_*}=0,\;j \in \III_{\mathrm{rpv}};\ (z,T) \geqslant 0\}
 \end{equation*}
 is compact.
\end{lemma}
\begin{proof}
 Case 1: isothermal combustion\newline
 The mass conservation together with the positivity and the fixed temperature
 define a polytope in  $\RR^{\nspec+1}$ which is closed and bounded,
 hence (Heine--Borel theorem) compact.\newline
 Case 2: adiabatic combustion\newline
 As in the isothermal case, the variables $z_i$ are restricted to a compact
 polytope due to elemental mass conservation and positivity constraints.
 Following Lemma~\ref{lemma:proper}, the preimage of the singleton of the fixed
 energy/enthalpy is a compact subset of $\RR^{\nspec+1}$. This subset may only
 be further constrained by the polytope defined by the mass conservation and
 positivity, and the intersection of compact subsets is compact.\qed
\end{proof}

\begin{lemma}[Existence of a solution]
 If the map $\Phi:\RR^{\nspec}\rightarrow\RR$ in the objective functional of the
 optimization problem $(\ref{eq:op})$ is a continuous function and the feasible
 set is not empty, there exists a solution of problem $(\ref{eq:op})$.
\end{lemma}
\begin{proof}
 Following the argumentation in \cite{Lebiedz2011a}, the semi-infinite
 optimization problem ($\ref{eq:op}$) can be reduced to a finite dimensional
 optimization problem by construction of a continuous map
 $(z,T)(t_*)\mapsto \int_{t_0}^{t_{\textrm f}}  \Phi\left(z(t)\right) \; \d t$.
 As seen in Lemma~\ref{lemma:compact}, the feasible set $M$ is compact.
 Therefore, existence follows from the Weierstra{\ss} theorem.\qed
\end{proof}

% ---------------------------------------------------------------------------- %
\section{Numerical methods}\label{s:numerics}
The semi-infinite optimization problem ($\ref{eq:op}$) can be solved after
suitable discretization of the ODE constraints e.g.\ either by a sequential
quadratic programming (SQP) \cite{Powell1978} or an Interior Point (IP) method,
see e.g.\ the review \cite{Forsgren2002}.

\subsection{Discretization}
In general, there are two ways for discretization and solution of
($\ref{eq:op}$): the sequential and the simultaneous approach.

\subsubsection{Sequential approach} In the sequential approach, ODE solution and
optimization are fully decoupled. The initial values $(z(t_0),T(t_0))$ are used
as optimization variables. Starting at $t_0$, the system is integrated with a
stiff ODE solver, e.g.\ via a backward differentiation formulae (BDF)
scheme \cite{Curtiss1952}. The integrand in
the objective function (\ref{eq:op:obj}) is integrated itself, and the end point
is evaluated in sense of a Mayer term objective functional. The optimization
iteration is performed after that based on the results of the integration and
computed derivative information. In the cases $t_* = t_0$ and $t_* = t_\f$, a
single shooting is appropriate as we deal with a stable ODE system; whereas in
case of $t_* \in (t_0,t_\f)$, a double shooting is needed for the values at
$t_*$.

\subsubsection{Simultaneous approach} Sometimes (e.g.\ in case of unstable or
extremely stiff systems), it is beneficial to use an all-at-once approach using
collocation formulae \cite{Ascher1998}. In this simultaneous approach, the
solution of the dynamic constraints and the optimization are coupled. The
interval $[t_0,t_\f]$ is divided into sub-intervals. Via e.g.\ a collocation
method, polynomials are constructed on each sub-interval tangent to the vector
field of the dynamics (\ref{eq:op:dyn_m}) and (\ref{eq:op:dyn_e}) approximating
their solution, and the corresponding formulae are treated as constraints in the
optimization iteration. In a collocation approach, we use a Gau\ss-Radau formula
with linear, quadratic, and cubic polynomials, respectively, because they have
stiff decay \cite{Ascher1998}.

\subsection{Solution of the finite-dimensional optimization problem}
In both cases (sequential and simultaneous), the result of the discretization is
a finite-dimensional nonlinear programming (NLP) problem. This can be solved
using an SQP algorithm or an IP method. The SQP algorithm treats the inequality
constraints using an active set strategy, see e.g.\ \cite{Nocedal2006}. Newton's
method is applied to the first order optimality conditions of a quadratic
approximation of the NLP problem including only equality and active inequality
constraints. By activating and deactivating constraints, the active set in the
solution is identified. In contrast in an IP method, the inequality constraints
are coupled to the objective function via a barrier term forcing the iterates
into the interior of the feasible domain. The resulting equality constrained NLP
problem is solved with a homotopy method: Newton's method is applied to the
first order optimality conditions. In the progress of optimization, the barrier
parameter is driven to zero to follow an homotopy path to the solution of the
NLP problem.

\subsection{Algorithms and software}
In Section \ref{s:results}, we present results of an application of the model
reduction method. We use IPOPT \cite{Waechter2006} as the main optimization
tool. It turned out to be a robust IP algorithm appropriate for our problems.
For the solution of linear equation systems within the optimization algorithm,
HSL routines \cite{HSL2007} and MUMPS \cite{Amestoy2001}, resp., are used.
Derivatives needed for the optimization are computed with the open source
automatic differentiation package CppAD \cite{Bell2010,Bell2008}.

For discretization of the optimization problem, we use a collocation approach
based on a Gau\ss-Radau method \cite{Ascher1998}. Alternatively, we use a
shooting approach including a BDF integrator that has been developed by
D.~Skanda for \cite{Skanda2012}.
For the numerical solution strategies, see also \cite{Lebiedz2012}.
We use MATLAB for plotting.

% ---------------------------------------------------------------------------- %
\section{Results}\label{s:results}
In this section, results for the application of the optimization based model
reduction method are shown. As this manuscript is focused on the application to
realistic systems, we skip a discussion of test equations for model reduction
methods. These can be found in \cite{Lebiedz2010,Lebiedz2011a}. We only consider
combustion mechanisms providing the complete kinetic and thermodynamic data.

In \cite{Lebiedz2011a}, it has been shown that the reverse mode ($t_*=t_\f$) of
the method identifies the correct SIM in case of a linear test model and the
Davis--Skodje test model \cite{Davis1999}, which has an analytically given
one-dimensional SIM, for infinite time horizon $t_0 \rightarrow -\infty$. Hence
we use the reverse mode for all results presented here.

\subsection{Ozone decomposition}\label{ss:o3}
As a first small test problem, we consider an ozone decomposition mechanism
including only three allotropes of oxygen, namely atomic oxygen, dioxygen and
ozone. The mechanism is given in Table \ref{t:O3}.
\begin{table}
\caption
 {\label{t:O3}Ozone decomposition mechanism for the forward rates as in
 \cite{Maas1989}. Collision efficiencies in reactions including M:
 $\alpha_{\species{O}}=1.14, \alpha_{\species{O_2}} = 0.40, \alpha_{\species{O_3}} = 0.92$.}
\begin{tabular}{lclrrr}
\hline\noalign{\smallskip}
 Reaction &&& $A$ / $\textrm{cm}, \textrm{mol}, \textrm{s}$ & $b$ & $E_\textrm{a}$ / $\si{\kilo\joule\per\mole}$ \\
\noalign{\smallskip}\hline\noalign{\smallskip}
 $\species{O} + \species{O} + \species{M}$ & $\rightleftharpoons$ & $\species{O_2} + \species{M}$ &  $2.9\times 10^{17}$ & $-1.0$ & $0.0$  \\
 $\species{O_3} + \species{M}$ & $\rightleftharpoons$ & $\species{O} + \species{O_2} + \species{M}$ & $9.5\times 10^{14}$ & $0.0$ &  $95.0$ \\
 $\species{O} + \species{O_3}$ & $\rightleftharpoons$ & $\species{O_2} + \species{O_2}$ & $5.2\times 10^{12}$ & $0.0$ & $17.4$\\
\noalign{\smallskip}\hline
\end{tabular}
\end{table}

The thermodynamic data is used in form of NASA polynomial coefficients. In
comparison to the results in \cite{Lebiedz2010}, we set up the mechanism in the
framework described in Section \ref{s:chemistry}. Reverse rate coefficients are
derived from equilibrium thermodynamics. For mass conservation, the elemental
specific mole has to be given which is
\begin{equation*}
\bar{z}_\species{O}= \frac{1000}{\SI{15.999}{\gram\per\mole}} = \SI{62.5}{\mole\per\kilogram}.
\end{equation*}
We consider a density of
$\rho = \SI{0.2}{\kilogram\per\cubic\meter}$ in the isochoric case and a
pressure of $p = \SI{e+5}{\pascal}$ for isobaric conditions, respectively.

In case of the ozone mechanism (Table \ref{t:O3}), it is technically not
necessary to demand positiveness of specific moles and temperature because
no pressure dependent reactions are present. Therefore, the mechanism can
be evaluated also in nonphysical regions (negative species concentrations
visible in the plots).

Results for the four different thermodynamic environments are shown in
Figure~\ref{f:O3-c1}, \ref{f:O3-c2}, \ref{f:O3-c3}, and \ref{f:O3-c4}. The model
has two degrees of freedom; we compute a numerical approximation of a
one-dimensional SIM. The (blue) open rings in the plots are the solution points
$(z,T)(t_*)$. Orbits through these points in forward and reverse direction are
also shown (blue curves), where the reverse part coincides with the optimal
trajectory piece $(z,T)(t)$, $t\in[t_0,t_\f]$. The trajectories converge toward
equilibrium which is shown as full (red) dot on the left hand side of the
subfigures.

\begin{figure}
 \includegraphics[width=\textwidth]{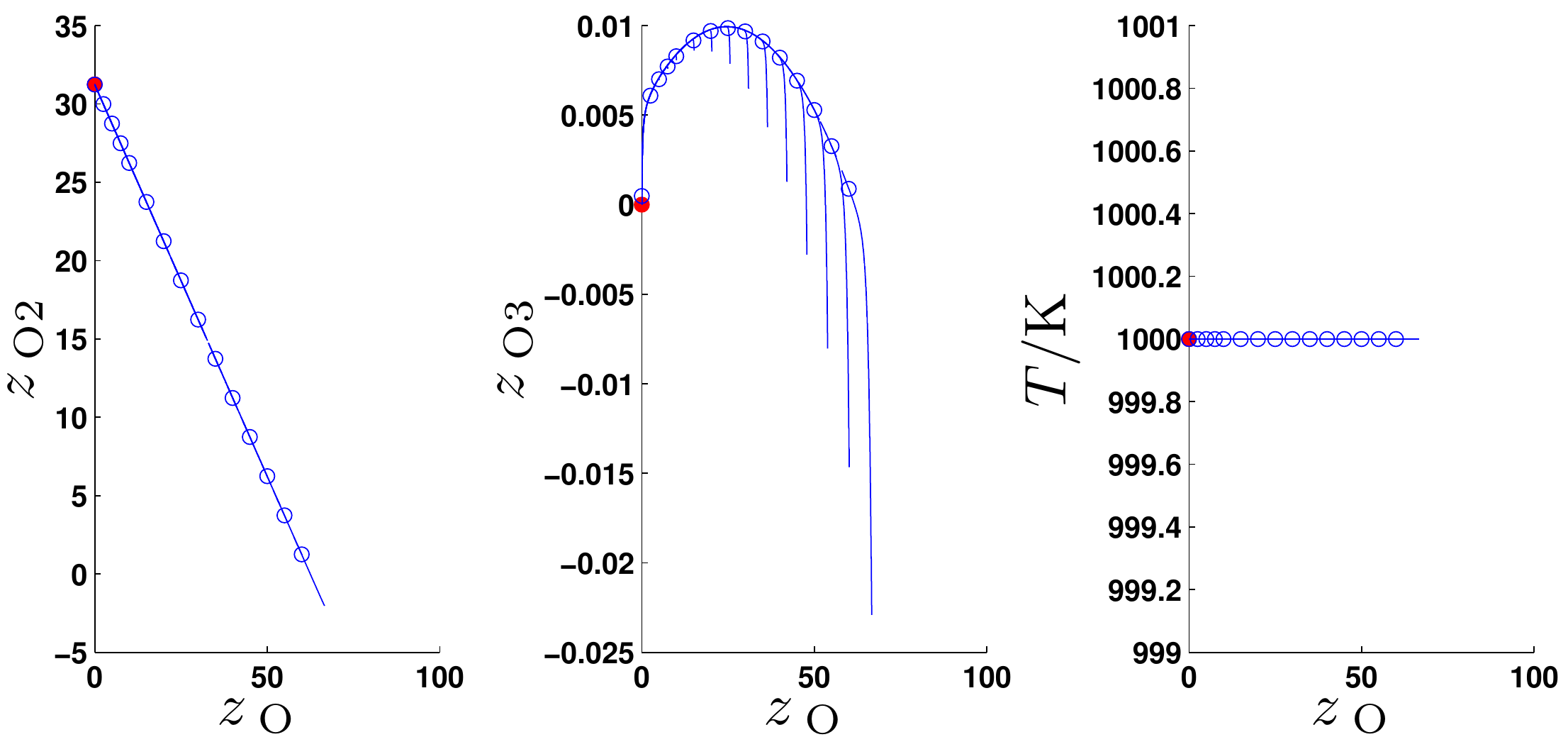}
 \caption{\label{f:O3-c1}Results (one-dimensional SIM) for the ozone
 decomposition mechanism modeled within an isothermal and isochoric environment.
 The value of $z_\species{O}$ serves as reaction progress variable. The
 optimization problem is solved several time for different values of
 $z_\species{O}^{t_*}$ varying between zero and the largest possible value
 $\bar{z}_\species{O}$. We use the reverse mode ($t_\f=t_*$) with an integration
 interval of $t_\f-t_0 = \SI{e-6}{\second}$. The free $z_i$
 (in $\si{\mole/\kilogram}$) and temperature are plotted versus
 $z_\species{O}$.}
\end{figure}
\begin{figure}
 \includegraphics[width=\textwidth]{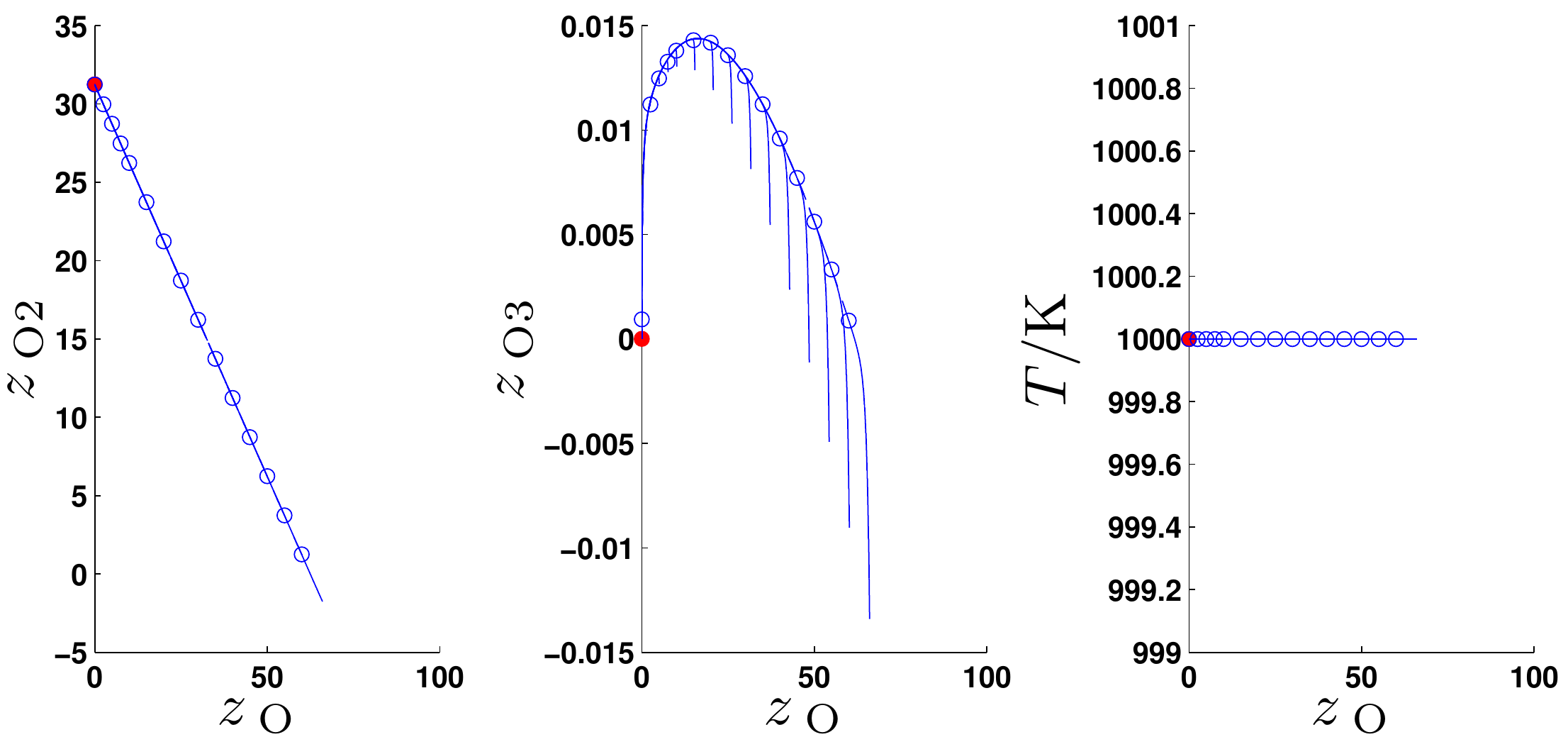}
 \caption{\label{f:O3-c2}Results (one-dimensional SIM) for the ozone
 decomposition mechanism modeled within an isothermal and isobaric environment.
 The plot is arranged as in Figure \ref{f:O3-c1}. Again we use the reverse
 mode ($t_\f=t_*$) with an integration interval of
 $t_\f-t_0 = \SI{e-6}{\second}$.}
\end{figure}
\begin{figure}
 \includegraphics[width=\textwidth]{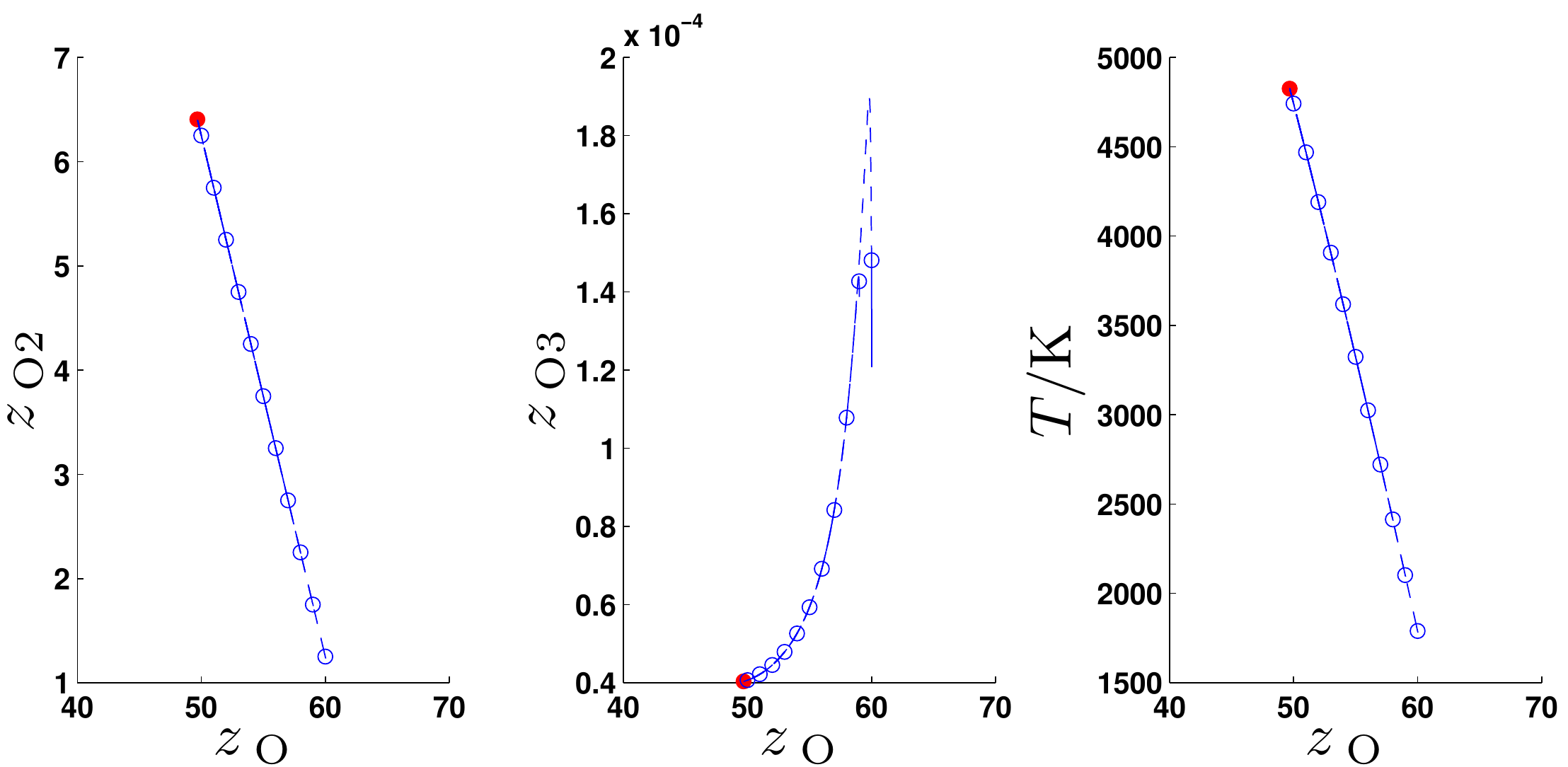}
 \caption{\label{f:O3-c3}Results (one-dimensional SIM) for the ozone
 decomposition mechanism modeled within an adiabatic and isochoric environment.
 The plot is arranged as in Figure \ref{f:O3-c1}. As in the adiabatic case, the
 reaction evolves faster the integration horizon is reduced to
 $t_\f-t_0 = \SI{e-8}{\second}$.}
\end{figure}
\begin{figure}
 \includegraphics[width=\textwidth]{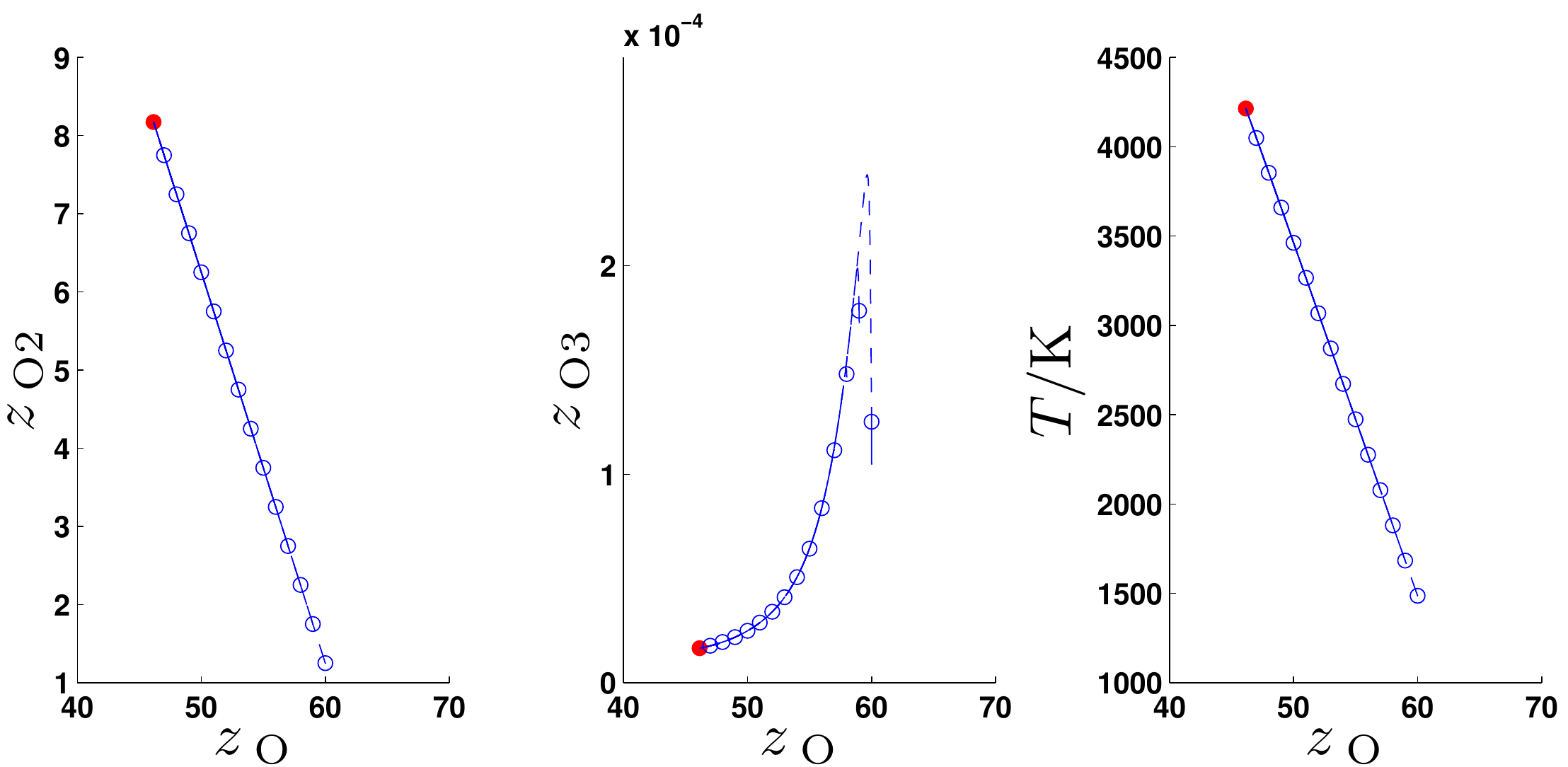}
 \caption{\label{f:O3-c4}Results (one-dimensional SIM) for the ozone
 decomposition mechanism modeled within an adiabatic and isobaric environment.
 The plot is arranged as in Figure \ref{f:O3-c1}. Again the integration horizon
 is $t_\f-t_0 = \SI{e-8}{\second}$.}
\end{figure}

In all plots (Figure \ref{f:O3-c1}--\ref{f:O3-c4}), it can be seen that near
equilibrium very good results can be achieved as the SIM approximation
is nearly invariant, i.e.\ all open dots are lying along one (slow) trajectory.
But far from equilibrium at large values of the reaction progress variable 
$z_\species{O}$, the full dynamics are active, so the invariance of the SIM is
poor due to a lack of time scale separation. A short relaxation phase can be
stated, but at least the values are in a reasonable range.

\subsection{Simplified hydrogen combustion mechanism}
In this section, we review a test case for a simplified combustion mechanism. The
presentation and results are similar to those presented in \cite{Lebiedz2011a}.

The reaction mechanism is given in Table \ref{t:ICE}. In \cite{Lebiedz2011a}, we
show a comparison to the results of Al-Khateeb et al.\ in
\cite{Al-Khateeb2009,Al-Khateeb2010}. Hence we use the thermodynamical data
(in form of NASA coefficients) we
received from J.~M.~Powers and A.~N.~Al-Khateeb. The mechanism itself has been
published originally in \cite{Li2004}. The simplified version shown in
Table~\ref{t:ICE} has been used by Ren et al.\ in \cite{Ren2006a}. The mechanism
itself consists of five reactive species and inert nitrogen, where in
comparison to a full hydrogen combustion mechanism the species $\species{O}_2$,
$\species{HO}_2$, and $\species{H}_2\species{O}_2$ are removed. The species
are involved in six Arrhenius type reactions, where three
combination/decomposition reactions require a third body for an effective
collision.

\begin{table}
 \caption
  {\label{t:ICE}The simplified mechanism as used in \cite{Ren2006a}. Collision efficiencies $\species{M}$:
  $\alpha_{\species{H}} = 1.0, \alpha_{\species{H_2}} = 2.5, \alpha_{\species{OH}} = 1.0$,
  $\alpha_{\species{O}} = 1.0, \alpha_{\species{H_2O}} = 12.0, \alpha_{\species{N_2}} = 1.0$.}
\begin{tabular}{lclrrr}
\hline\noalign{\smallskip}
 Reaction &&& $A$ / ($\textrm{cm}, \textrm{mol}, \textrm{s}$) & $b$ & $E_\textrm{a}$ / $\si{\kilo\joule\per\mole}$ \\
\noalign{\smallskip}\hline\noalign{\smallskip}
 $\species{O} + \species{H_2} $ & $\rightleftharpoons$ & $\species{H} + \species{OH}$ &  $5.08\times 10^{04}$ & $2.7$ & $26.317$  \\
 $\species{H_2} + \species{OH}$ & $\rightleftharpoons$ & $\species{H_2O} + \species{H}$ & $2.16\times 10^{08}$ & $1.5$ & $14.351$ \\
 $\species{O}$ + $\species{H_2O}$ & $\rightleftharpoons$ & $\species{2}\,\species{OH}$ & $2.97\times 10^{06}$ & $2.0$ & $56.066$ \\
 $\species{H_2} + \species{M}$ & $\rightleftharpoons$ & $\species{2}\,\species{H} + \species{M}$ & $4.58\times 10^{19}$ & $-1.4$ & $436.726$ \\
 $\species{O} + \species{H} + \species{M}$ & $\rightleftharpoons$ & $\species{OH} + \species{M}$ & $4.71\times 10^{18}$ & $-1.0$ & $0.000$ \\
 $\species{H} + \species{OH} + \species{M}$ & $\rightleftharpoons$ & $\species{H_2O} + \species{M}$ & $3.80\times 10^{22}$ & $-2.0$ & $0.000$ \\
\noalign{\smallskip}\hline
\end{tabular}
\end{table}

Al-Khateeb et al.\ identified a one-dimensional SIM for a model including this
mechanism \cite{Al-Khateeb2009}. The model additionally involves the following
parameters: The combustion is considered in an isothermal and isobaric
environment with a temperature of $T=\SI{3000}{\kelvin}$ and a pressure
of $p=\SI{101325}{\pascal}$. The elemental mass conservation is given in
terms of amount of species; it is
\begin{equation*}
 \begin{aligned} \label{eq:conserve}
  n_\mathrm{H}+2\,n_{\mathrm{H}_2}+n_\mathrm{OH}+2\,n_{\mathrm{H}_2\mathrm{O}} &= \SI{1.25e-3}{\mole} \\
  n_\mathrm{OH}+n_\mathrm{O}+n_{\mathrm{H}_2\mathrm{O}} &= \SI{4.15e-4}{\mole} \\
  2\,n_{\mathrm{N}_2} &= \SI{6.64e-3}{\mole}.
 \end{aligned}
\end{equation*}
Therefore, the total mass in the system is $m=\SI{1.01e-4}{\kilogram}$. We
continue to use the specific moles as our standard variables here and use
$z_i = \tfrac{n_i}{m}$ in the following.

The results are shown in Figure \ref{f:Li_mech_3d}. There is a very good
agreement of our results with theirs. Even on both sides of equilibrium,
our approximations coincide with the correct one-dimensional
SIM on its both branches which consist of two heteroclinic orbits in this case.
\begin{figure}
 \includegraphics[width=\textwidth]{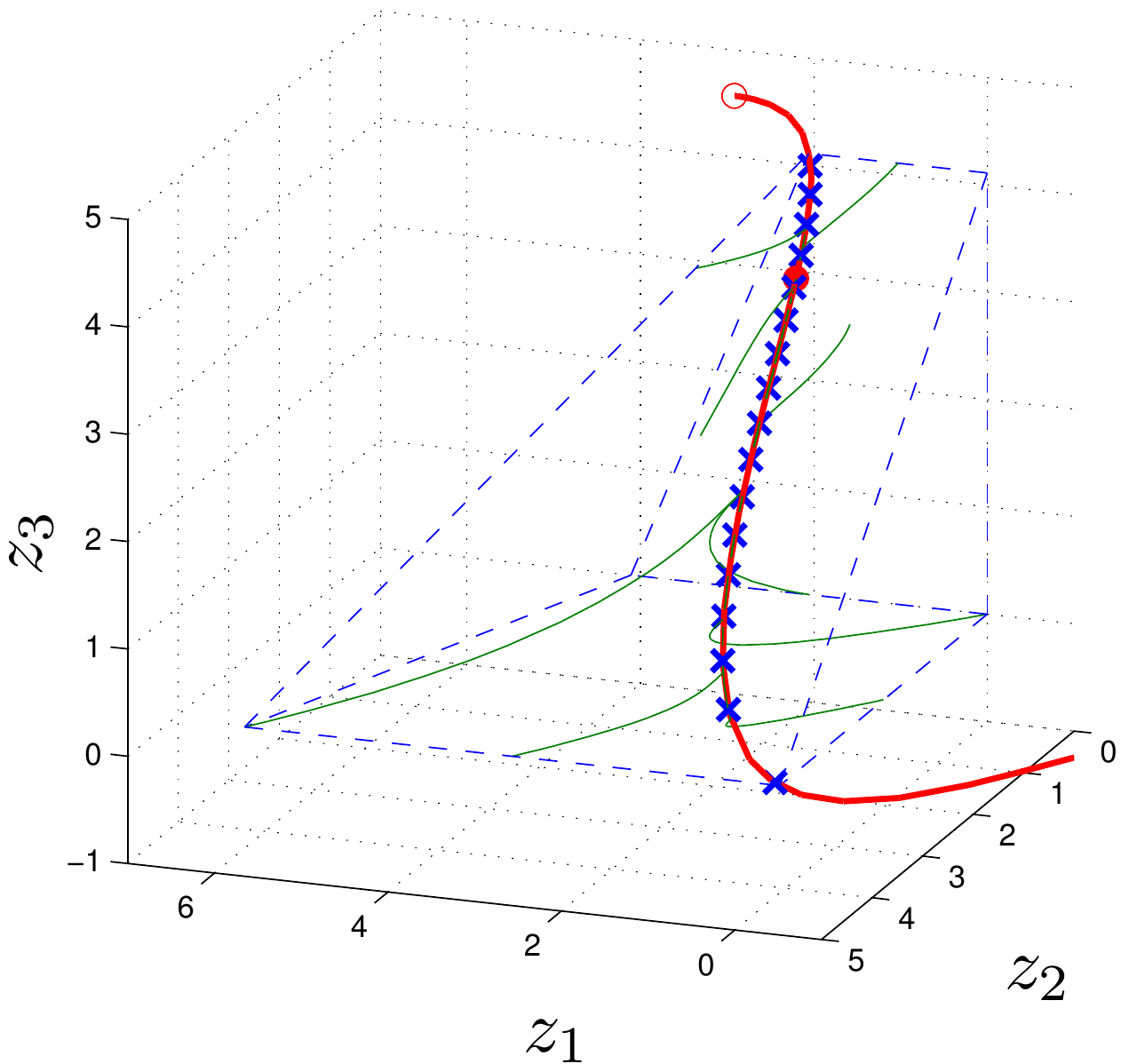}
 \caption{\label{f:Li_mech_3d}Three-dimensional plot of the results for the numerical
 approximation of a one-dimensional SIM of the simplified combustion mechanism
 computed with the reverse mode, $z_{\mathrm{H}_2\mathrm{O}}$ as reaction
 progress variable, and $t_\f - t_0 = \SI{e-7}{\second}$.
 The illustration is similar to Figure 9 in \cite{Al-Khateeb2009}. The blue
 bounded polytope shows the physically feasible state space. Green curves
 correspond to some ``arbitrary'' trajectories for illustration. The red
 curve depicts the two branches of the SIM as derived in \cite{Al-Khateeb2009}.
 The open red dot represents the unstable fixed point ($R_6$ in
 \cite{Al-Khateeb2009}); the full red dot represents the equilibrium ($R_7$).
 Our results are included as blue crosses. The state of the species is given as
 $z_1=z_{\species{H}_2}$, $z_2=z_{\species{O}}$, and
 $z_3=z_{\species{H}_2\species{O}}$ in $\si{\mole\per\kilogram}$.}
\end{figure}

\subsection{Syngas combustion mechanism}\label{ss:syngas}
As a last example of a full detailed chemistry combustion mechanism, we use a
syngas combustion extracted from the GRI 3.0 mechanism \cite{Smith1999}. It
consists of all 33 reactions of the GRI 3.0 mechanism which involve no other
species than $\species{O}$, $\species{O}_2$, $\species{H}$, $\species{OH}$,
$\species{H}_2$, $\species{HO}_2$, $\species{H}_2\species{O}_2$,
$\species{H}_2\species{O}$, $\species{N}_2$, $\species{CO}$, and
$\species{CO}_2$. Those 33 reactions can be split up into 31 Arrhenius-type and
two pressure-dependent ones.

The overall reaction can be stated as
\begin{equation*}
 \species{H}_2 + \species{CO} + \species{O}_2 \rightarrow \species{H}_2\species{O} + \species{CO}_2,
\end{equation*}
where $\species{N}_2$ only serves as a collision partner. We assume a
stoichiometric mixture of syngas with air in an adiabatic and isochoric
environment. As fixed mass density, we use
$\rho = \SI{0.3}{\kilogram\per\cubic\meter}$.
We assume a ratio of $n_{\species{H}_2} : n_{\species{CO}} = 1 : 1$, and a ratio
of $n_{\species{O}_2} : n_{\species{N}_2} = 1 : 3.76$. This leads to a unburned
mixture of
$z_{\species{CO}} = z_{\species{H}_2} = \SI{5.973}{\mole\per\kilogram}$ and
$z_{\species{N}_2} = \SI{22.46}{\mole\per\kilogram}$. The specific internal
energy of this mixture at a temperature of $T=\SI{1000}{\kelvin}$ is used as a
fixed specific internal energy for SIM computation.

Results for this model are shown in Figure \ref{f:Syngas}.
\begin{figure}
 \includegraphics[width=\textwidth]{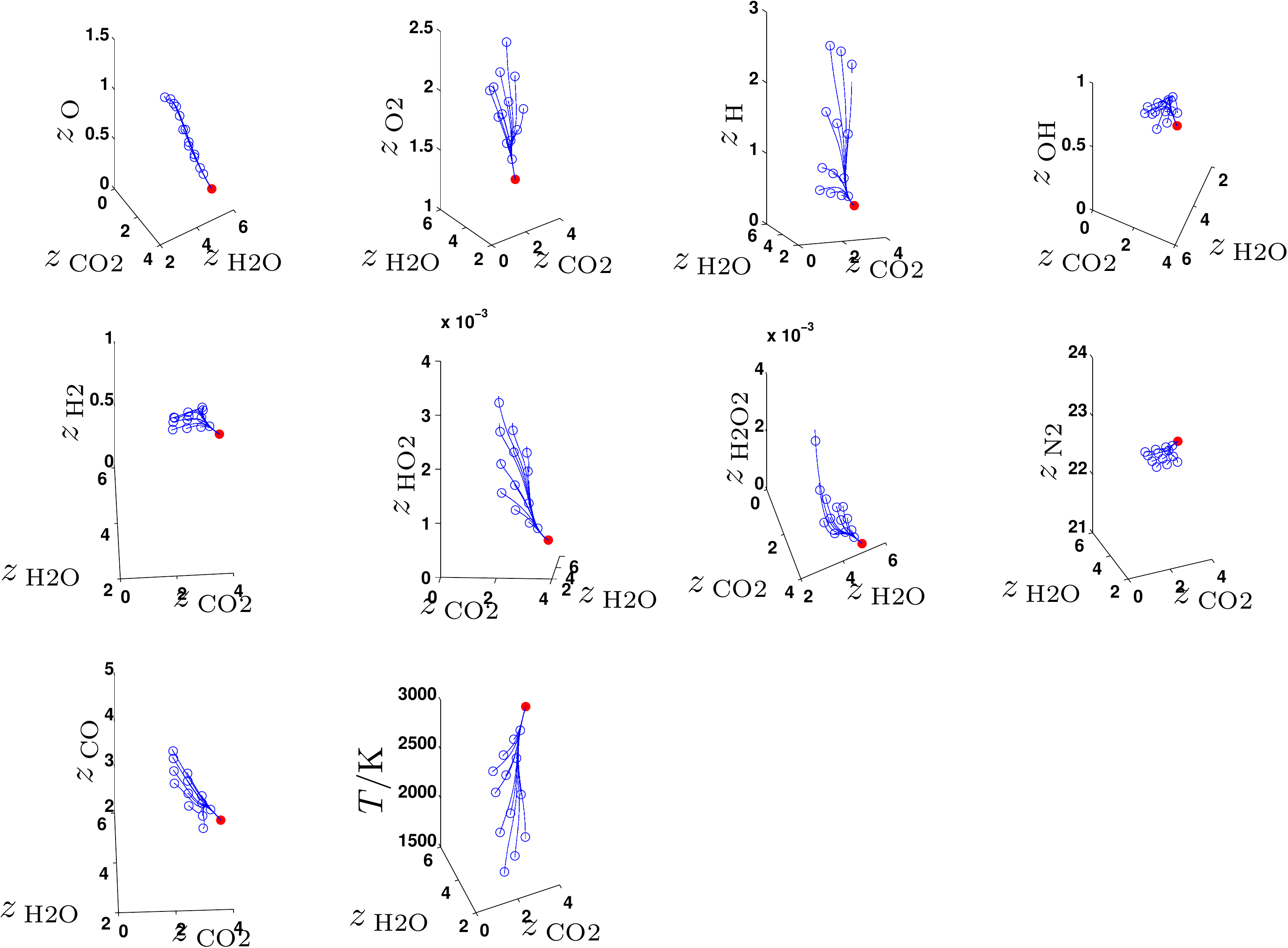}
 \caption{\label{f:Syngas}Visualization of the results (two-dimensional SIM) of
 the model reduction method applied to the syngas combustion model described in
 Section~\ref{ss:syngas}. The reverse mode is applied with a time horizon of
 $\SI{e-7}{\second}$. We approximate a two-dimensional manifold and
 use the overall products $z_{\species{H}_2\species{O}}$ and
 $z_{\species{CO}_2}$ as reaction progress variables.}
\end{figure}
The same style as in Figure \ref{f:O3-c2} is used. The resulting SIM
approximation points (solutions of the optimization problem (\ref{eq:op})) are
shown as open blue dots together with trajectories emanating from these and
converging to equilibrium, which is shown as full red dot. In the
two-dimensional case, invariance cannot be seen by eye inspection, but a
reasonable manifold is found.

% ---------------------------------------------------------------------------- %
\section{Conclusions}\label{s:summary}
An optimization method is presented that allows for efficient model reduction of
realistic combustion models. It is applied to three models for testing its
applicability. It can be seen that for an appropriate time scale separation the
solution of an optimization problem approximates points on a slow invariant
manifold. The application of the model reduction approach to a realistic syngas
combustion model considered in an adiabatic and isochoric environment
demonstrates the applicability to realistic large scale mechanisms.

Further research will be needed for an identification of an appropriate number
and choice of the reaction progress variables for large scale mechanisms.

% ---------------------------------------------------------------------------- %
\begin{acknowledgements}
This work was supported by the German Research Foundation (DFG) via
project B2 within the Collaborative Research Center (SFB) 568.

The authors wish to thank the late J\"urgen~Warnatz (IWR, Heidelberg) for
providing professional mentoring for combustion research. The authors also thank
Markus~Nullmeier (IWR, Heidelberg) for his helpful feedback.
\end{acknowledgements}


% bib file
%\bibliography{./literature.bib}   % name your BibTeX data base

% hard coded bbl

\end{document}